\documentclass[copyright]{eptcs}

\providecommand{\event}{E. Formenti (Ed.): AUTOMATA \& JAC 2012} 

\usepackage{graphicx}
\usepackage{mathrsfs}
\usepackage{amssymb}
\usepackage{amsmath}
\usepackage{theorem}

\newcommand{\tfam}{\mathscr{T}}

\newcommand{\ca}{\textrm{CA}}
\newcommand{\cat}{\textrm{CAT}}
\newcommand{\ia}{\textrm{IA}}
\newcommand{\iat}{\textrm{IAT}}
\newcommand{\iatrtrt}{\textrm{IAT}_{rt,rt}}
\newcommand{\iatrtlt}{\textrm{IAT}_{rt,lt}}
\newcommand{\iatltlt}{\textrm{IAT}_{lt,lt}}
\newcommand{\iatltrt}{\textrm{IAT}_{lt,rt}}
\newcommand{\catrtrt}{\textrm{CAT}_{rt,rt}}
\newcommand{\catrtlt}{\textrm{CAT}_{rt,lt}}
\newcommand{\catltlt}{\textrm{CAT}_{lt,lt}}
\newcommand{\catltrt}{\textrm{CAT}_{lt,rt}}
\newcommand{\dfst}{\textrm{DFST}}
\newcommand{\ufst}{\textrm{UFST}}
\newcommand{\sfst}{\textrm{SFST}}
\newcommand{\dpdt}{\textrm{DPDT}}
\newcommand{\spdt}{\textrm{SPDT}}
\newcommand{\updt}{\textrm{UPDT}}

\newcommand{\border}{\texttt{\#}}
\newcommand{\dollar}{\texttt{\$}}
\newcommand{\kand}{\texttt{\&}}
\newcommand{\rightend}{\mathord{\vartriangleleft}}  

\theorembodyfont{\slshape} 
\theoremstyle{plain} 
 \newtheorem{definition}{Definition}
 \newtheorem{lemma}[definition]{Lemma}
 \newtheorem{theorem}[definition]{Theorem}
 \newtheorem{corollary}[definition]{Corollary}
 \newtheorem{proposition}[definition]{Proposition}

\theorembodyfont{\normalfont}
 \newtheorem{example}[definition]{Example}

\def\squareforqed{$\Box$}
\def\qed{\ifmmode\squareforqed\else{\unskip\nobreak\hfil%
  \penalty50\hskip1em\null\nobreak\hfil\squareforqed%
  \parfillskip=0pt\finalhyphendemerits=0\endgraf}\fi}

\newenvironment{proof}{\noindent{\textbf{Proof}\ }}{\qed\medskip}

\title{Transductions Computed by\\ One-Dimensional Cellular Automata}

\def\titlerunning{Cellular Automaton Transducers}
\def\authorrunning{M.~Kutrib, A.~Malcher}

\author{%
Martin Kutrib and Andreas Malcher
\institute{Institut f\"ur Informatik, Universit\"at Giessen,\\
  Arndtstr.~2, 35392 Giessen, Germany}
  \email{$\{$kutrib,malcher$\}$@informatik.uni-giessen.de}
}

\begin{document}

\maketitle

\begin{abstract}
Cellular automata are investigated towards their ability to compute 
transductions, that is, to transform inputs into outputs.
The families of transductions computed are classified with regard to 
the time allowed to process the input and to compute the output.
Since there is a particular interest in fast transductions,
we mainly focus on the time complexities real time and linear time.
We first investigate the computational capabilities of
cellular automaton transducers by comparing them to 
iterative array transducers, that is, we compare parallel
input/output mode to sequential input/output mode of massively parallel
machines.
By direct simulations, it turns out that the parallel mode 
is not weaker than the sequential one.
Moreover, with regard to certain time complexities cellular
automaton transducers are even more powerful than iterative arrays.
In the second part of the paper, the model in question is compared 
with the sequential devices single-valued finite state transducers 
and deterministic pushdown transducers. It turns out that both models 
can be simulated by cellular automaton transducers faster than by iterative 
array transducers.
\end{abstract}

\section{Introduction}\label{sec:intro}

Cellular automata have widely been investigated as a massively parallel
computation model. In connection with the problems to recognize syntactical
patterns, or in our terms formal languages, cellular automata have been
considered in~\cite{Smith:1970:CAFL,Smith:1972:RTLRODCA} for the first time.
Over the years substantial progress has been achieved in this field
but there are still some basic open problems with deep relations 
to other fields (see, for example,~\cite{kutrib:2009:calt}).
The results comprise the relations between parallel and sequential input
mode, the impact of two-way and one-way inter-cell communication, 
the relation between different time complexities such as real time, linear time, 
or arbitrary time, the capabilities with respect to accept linguistic
language families such as regular and context-free languages, closure properties, 
and decidability questions.

Computational models are not only interesting from the viewpoint of
recognizing some input, but also from the viewpoint of transforming an input into
an output. For example, a parser for a formal language should not only return the
information whether or not the input word can be parsed, but also the parse tree
in the positive case. Here, two things are important: first the information
whether the input is accepted, i.e., whether the input can be parsed, and second
the production of the output, i.e., the construction of a parse tree if the
word can be parsed. If the word cannot be parsed, the output is no interest.
This motivation led to the investigation of models such as finite state transducers or
pushdown transducers which are classical finite automata or pushdown
automata where each transition is associated with some output. If an input
is accepted, the result of the transduction is the output of the transitions 
in the order they have been applied. Both models have been studied in detail 
(see, for example,~\cite{Aho:1972:tptcv1,Berstel:1979:tcfl}),
and many applications such as in the context of parsing are known.

Parallel transducers have been investigated in~\cite{Ibarra:1991:PPOWLAFSM,kutrib:2010:tcia:proc}.
In the first paper, one-way linear iterative arrays are introduced. In
this model, the input is supplied to the leftmost cell, and the output is emitted at the
rightmost cell. In between there are as many cells as the input is long, the information
flow is one-way, from left to right, and the output does not depend on the
fact whether or not the input is accepted. In the latter paper, iterative
array transducers are introduced, where the leftmost cell receives the input
and emits the output. For the computational capacity of such devices, the time 
complexities for processing the input as well as for computing the output
are crucial. It turned out that iterative array transducers with
time complexity (real-time, real-time) are less powerful than those 
working in (real-time, linear-time). In turn, the latter are less powerful 
than iterative array transducers with time complexity (linear-time, linear-time).
Moreover, it has been shown that deterministic finite state transducers 
as well as certain deterministic pushdown transducers can be simulated by 
iterative arrays with time complexity (real-time, real-time), whereas 
nondeterministic single-valued finite state transducers are simulated by 
iterative arrays in (real-time, linear-time).

Here, we complement these results with investigations of the computational
capacity of cellular automaton transducers compared with iterative arrays 
and sequential devices. The detailed definitions and two examples are 
given in Section~\ref{sec:def}. The results on the computational
capacities compared with iterative array transducers are obtained 
in Section~\ref{sec:compcap}. First, we present a construction which
shows that any iterative array transducer can be simulated by some
cellular automaton transducer preserving the time complexity as long as 
the time complexity is `fair'. Together with the example of a 
transduction not computed by any iterative array in real time,
we obtain as a consequence that cellular automaton transducers are 
more powerful than iterative array transducer with respect
to the time complexities (real-time, real-time) and (real-time, linear-time). For the
combination (linear-time, linear-time) both devices are shown to be equally powerful.
In Section~\ref{sec:comptransd}, we compare the model in question with 
sequential machines. In particular, it is shown they can simulate 
single-valued finite state transducers and deterministic pushdown 
transducers faster than iterative arrays. The former are simulated in
(real-time, real-time) whereas the latter are simulated in (real-time, linear-time).

\section{Preliminaries and Definitions}\label{sec:def}

We denote the rational numbers by ${\mathbb{Q}}$, and 
the non-negative integers by $\mathbb{N}$. 
For the empty word we write~$\lambda$, the reversal of
a word $w$ is denoted by~$w^R$, and for the length of $w$
we write~$|w|$. 
For the number of occurrences of a symbol~$a$ in~$w$ we
use the notation~$|w|_a$. The set of all words over the alphabet $A$ 
whose lengths are at most $j\geq 0$ is denoted by $A^{\leq j}$. 
We write~$\subseteq$ for set inclusion, and $\subset$ for strict set inclusion.
In order to avoid technical overloading in writing, two 
languages~$L$ and $L'$ are considered to be equal, if they differ at most by
the empty word, that is, $L \setminus \{ \lambda \} = L' \setminus \{ \lambda \}$. 

A (one-dimensional) two-way cellular automaton transducer 
is a linear array of identical deterministic finite state machines, 
sometimes called cells, where each cell except the two outermost ones 
is connected to its both nearest neighbors.
We identify the cells by positive integers. 
The transition of a cell depends on its current state
and the current states of its neighbors, where the outermost cells 
receive information associated with a boundary symbol 
on their free input lines. The cells work synchronously at discrete time
steps.

The input/output mode for cellular automaton transducers is called parallel. 
One can suppose that all cells fetch their input symbol 
during a pre-initial step. Here we assume that each cell
is additionally equipped with an output register that is initially
empty and can be filled once by the cell. When all
output registers have been filled, the transduction 
is completed.

\begin{definition}
A \emph{cellular automaton transducer $(\cat)$} is a system
$\langle S,F,A,B,\border,\delta\rangle$, where 
\begin{enumerate}
\item
$S$ is the finite, nonempty set of \emph{cell states},
\item
$F\subseteq S$ is the set of \emph{accepting states},
\item
$A\subseteq S$ is the nonempty set of \emph{input symbols},
\item
$B$ is the finite set of \emph{output symbols} not including the special
  symbol $\bot$,
\item
$\border\notin S$ is the permanent \emph{boundary symbol}, and
\item
$\delta: (S\cup\{\border\})\times S\times (S\cup\{\border\}) \to S\times (B^*\cup\{\bot\})$ is the 
\emph{local transition function}. 
\end{enumerate}
\end{definition}

A \emph{configuration} of a cellular automaton transducer  
$M=\langle S,F,A,B,\border,\delta\rangle$
at time $t\geq 0$ is a description of its global state, which can formally be described
by two mappings $c_t:\{1,2,\dots,n\} \to S$ and
\mbox{$o_t:\{1,2,\dots,n\} \to B^*\cup\{\bot\}$,} for $n\geq 1$,
which map the single cells to their current states and to their output
emitted, where $\bot$ means no output so far. 
The operation starts at time~0 in a so-called \emph{initial configuration},
which is defined by the given input $w=a_1a_2\cdots a_n\in A^+$,
and no outputs. We set $c^{(w)}_{0}(i)=a_i$ and $o_0(i)=\bot$, for $1\leq i\leq n$. 
Successor configurations are computed
according to the global transition function~$\Delta$.
For convenience, we write $\delta_s$ for the projection on the first
component of $\delta$, that is, the successor state, and
$\delta_o$ for the projection on the second
component, that is, the output emitted.
Let $(c_t, o_t)$, $t\geq 0$, be a configuration with $n\geq 2$, then its 
successor $(c_{t+1}, o_{t+1})$ is as follows.
\[
c_{t+1}(i) =  
\begin{cases}
  \delta_s(c_t(i-1),c_t(i),c_t(i+1)) & \text{if } i\in\{2,3,\dots,n-1\}\\
  \delta_s(\border, c_t(1), c_t(2))  & \text{if } i=1\\
  \delta_s(c_t(n-1),c_t(n),\border)  & \text{if } i=n\\
\end{cases}
\]
For $n=1$, the next state of the sole cell is $\delta_s(\border,c_t(1),
\border)$. For $o_t$ we obtain
\[
o_{t+1}(i) =  
\begin{cases}
  o_t(i) & \text{if } o_t(i)\ne \bot\\
  \delta_o(c_t(i-1),c_t(i),c_t(i+1)) & \text{if } o_t(i)= \bot \text{ and }i\in\{2,3,\dots,n-1\}\\
  \delta_o(\border, c_t(1), c_t(2))  & \text{if } o_t(i)= \bot \text{ and } i=1\\
  \delta_o(c_t(n-1),c_t(n),\border)  & \text{if } o_t(i)= \bot \text{ and } i=n\\
\end{cases}
\]
and, thus, $\Delta$ is induced by~$\delta$.
As usual we extend $\Delta$ to sequences of configurations
and denote it by $\Delta^*$. That is, $\Delta^0$ is the identity,
$\Delta^t = \Delta(\Delta^{t-1})$, $1\leq t$, and
$\Delta^*=\bigcup_{0\leq t} \Delta^t$. Thus, 
$(c_t,o_t) \in \Delta^*(c,o)$ indicates that it is
possible for $M$ to go from the configuration $(c,o)$ to the
configuration $(c_t,o_t)$ in a sequence of zero or more steps.

An input $w$ is \emph{accepted} by a $\cat$ if at some time step
during the course of its computation the leftmost cell enters an 
accepting state. 
Transducer $M$ \emph{transforms} input words 
$w\in A^+$ into output words $v\in B^*$. For a successful transformation 
$M$ has to accept the input, otherwise the output is not recorded. Moreover,
each cell has to emit an output:
$$
M(w) = v,
$$
if $w$ is accepted by $M$, $(c_t, o_t)\in \Delta^*(c_0^{(w)}, o_0)$, 
$o_t(i)\neq \bot$, $1\leq i\leq |w|$, and $v=o_t(1)o_t(2)\cdots o_t(|w|)$.
The \emph{transduction realized by $M$}, denoted by $T(M)$,
is the set of pairs $(w,v)\in A^+\times B^*$ such that
$M(w) = v$.

Let $t_i, t_o:\mathbb{N}\to\mathbb{N}$ be two mappings. 
If for all $(w,v)\in T(M)$, the input $w$ is accepted after at most~$t_i(|w|)$ 
time steps, and $o_t(i)\neq \bot$, $1\leq i\leq |w|$, after at most 
$t_o(|w|)$ time steps, then~$M$ is said to be of time complexity $(t_i,t_o)$ and
we write $\cat_{t_i,t_o}$. 
The family of transductions realized by $\cat_{t_i,t_o}$ is denoted by 
$\tfam(\cat_{t_i,t_o})$. If $t_i$ or $t_o$ is the identity function~$n$,
we call it \emph{real time} and write $rt$. 
If $t_i(n)$ or $t_o(n)$ is of the form $k\cdot n$, for some 
$k\in\mathbb{Q}$, $k\geq 1$, we call it \emph{linear time} and write $lt$.

If we build the projection on the first components of $T(M)$,
then the cellular automaton transducer degenerates to a
cellular automaton acceptor ($\ca$). The projection on the first components 
is denoted by~$L(M)$.

In order to clarify the notation we give two examples. The first example
shows that $\cat$ can copy their input in real time. The second example shows
that $\cat$ can sort a binary input in real time.

\begin{example}\label{exa:copy}
The transduction $\{\,(w,ww) \mid w \in \{a,b\}^+\,\}$ belongs to
the family $\tfam(\catrtrt)$.
The basic idea is to use one register to shift the input from left to right
and another register to shift the input from right to left. Additionally,
two signals are started at both ends which cause the cells to emit the correct output
(see Figure~\ref{fig:copy}).

More detailed, let $w=a_1a_2 \cdots a_n$ be the input and $n$ be even. In the first time step,
cell $1$ emits $a_1a_2$ and cell $n$ emits $a_{n-1}a_n$. In the next
time step, cell $2$ emits $a_3a_4$ and cell $n-1$ emits $a_{n-3}a_{n-2}$. 
This is possible since the necessary information has been shifted to the 
corresponding cells and their neighborhood. Generalizing this observation
to $1 \le i \le \frac{n}{2}$, we obtain that cell $i$ emits $a_{2i-1}a_{2i}$
and cell $n-i+1$ emits $a_{n-2i+1}a_{n-2i+2}$
at time $i$. If $n$ is odd, the construction is similar. 
Now cells $i$ and $n-i+1$ emit $a_{2i-1}a_{2i}$ and $a_{n-2i+1}a_{n-2i+2}$
at time $i$, for $1 \le i \le \lfloor \frac{n}{2} \rfloor$. 
Finally, cell $\lceil \frac{n}{2} \rceil$ emits $a_na_1$ at time 
step~$\lceil \frac{n}{2} \rceil$.
\qed
\end{example}

\begin{figure}[!ht]
\centering
\includegraphics[scale=.9]{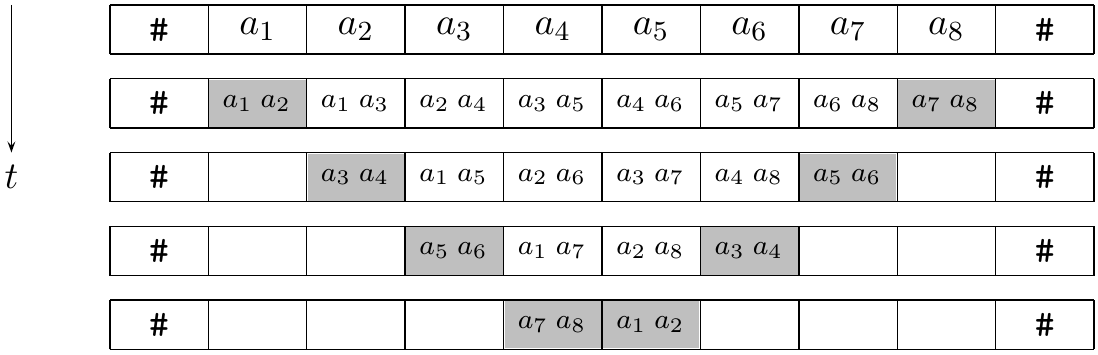}
\caption{Schematic computation of a $\catrtrt$ copying input $a_1a_2 \cdots a_8$.
The step at which a cell has emitted its output is highlighted.}
\label{fig:copy}
\end{figure}

\begin{example}\label{exa:sort}
Here, we consider the transduction $\{\,(w,a^{|w|_a}b^{|w|_b}) \mid w \in
\{a,b\}^+\,\}$ and show that it is computed by a $\catrtrt$.
The basic idea is that any two neighboring cells where the left cell
carries a $b$ and the right cell carries an $a$ switch their contents. 
By this local transpositions, all $a$'s will eventually be in the left 
part of the array whereas all $b$'s
will be in the right part (see Figure~\ref{fig:sort}).

Clearly, on any input of length $n$ the correct sorting has been achieved after at most $n$ time
steps. There is one problem to be solved: since the transpositions are only local,
a cell cannot know whether its current content is final and has to be
emitted. We can cope with the problem by synchronizing all cells at time~$n$. 
The synchronization is realized by the well-known Firing Squad Synchronization Problem (FSSP), 
that is implemented on an additional track.
In the first time step, two instances of the FSSP are started, one in the leftmost cell 
and the other one in the rightmost cell. First, let $n$ be even. When the
leading signals of both instances meet in the center at cells $\frac{n}{2}$ 
and $\frac{n}{2}+1$, they are reflected and cells $1,2, \ldots, \frac{n}{2}$ 
and cells $\frac{n}{2}+1, \frac{n}{2}+2, \ldots, n$ are treated as two
separate instances of the FSSP. Since each FSSP can be set up to synchronize 
$m$ cells in time $2m$, we obtain that $\frac{n}{2}$ cells are synchronized 
in time $n$. Since both FSSP work in parallel,
we obtain that all cells are synchronized at time step $n$. 
The remaining case where $n$ is odd is handled similarly.
\qed
\end{example}

\begin{figure}[!ht]
\centering
\includegraphics[scale=0.9]{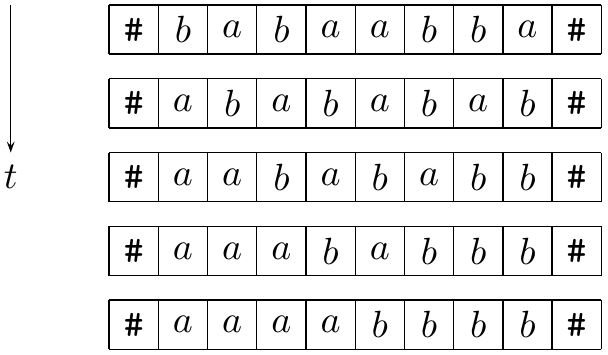}
\caption{Schematic computation of a $\catrtrt$ sorting an input $babaabba$ by local
transpositions. The final synchronization is not depicted.}
\label{fig:sort}
\end{figure}

\section{Computational Capacity of Cellular Automaton Transducers}
\label{sec:compcap}

This section is devoted to the computational capacity of cellular automaton
transducers compared with the power of iterative array transducers
as studied in~\cite{kutrib:2010:tcia:proc}. While $\cat$ are arrays of
finite state machines working in parallel input/output mode,
iterative array transducers receive the input and emit the output sequentially
by the leftmost cell, the so-called communication cell.

Basically, an iterative array transducer ($\iat$) is a linear 
array of deterministic finite state machines, where each
cell except the leftmost one is connected to its both nearest neighbors.
The distinguished leftmost cell is the communication cell that
is connected to its neighbor to the right and to the input/output supply
(see Figure~\ref{fig:pic-ia}). 
Initially, all cells are in the so-called quiescent state. 
At each time step the communication cell reads an input symbol
and writes a possibly empty string of output symbols.
To this end, we have different local transition functions. 
All cells but the communication cell change their state depending on
their current state and the current states of their neighbors. 
The state transition and output of the communication
cell depends on its current state, the current state of its neighbor, and 
on the current input symbol (or if the whole input has been consumed 
on a special end-of-input symbol). 
As for cellular automaton transducers, the cells work synchronously at 
discrete time steps.

\begin{figure}[!ht]
\centering
\includegraphics[scale=1]{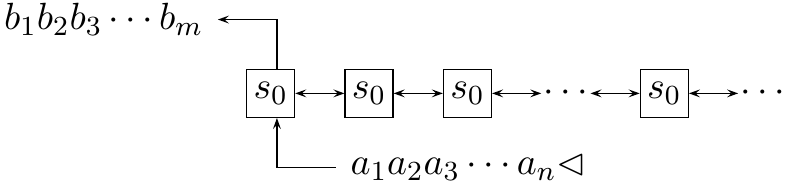}
\caption{An iterative array transducer.}
\label{fig:pic-ia}
\end{figure}

For details of later constructions, we provide the definition more formally.
An \emph{iterative array transducer} $(\iat)$ is a system 
$\langle S, F, A, B, \rightend, s_0, \delta,\delta_0 \rangle$, where 
$S$ is the finite, nonempty set of \emph{cell states},
$F\subseteq S$ is the set of \emph{accepting states},
$A$ is the finite set of \emph{input symbols},
$B$ is the finite set of \emph{output symbols},
$\rightend\notin A$ is the \emph{end-of-input symbol},
$s_0\in S$ is the \emph{quiescent state},
$\delta: S^3 \to S$ is the \emph{total local transition function for 
           non-communication cells} satisfying
$\delta(s_0,s_0,s_0)=s_0$, and
$\delta_0:(A\cup \{\rightend\})\times S^2\to B^*\times S$ is the 
\emph{partial local transition function for the communication cell}.
The $\iat$ halts when the transition function~$\delta_0$ is not defined
for the current configuration. Similar as for $\cat$, an input $w$ is accepted 
when the communication cell enters an accepting state at some time $t$
during the computation. The pair~$(w,v)$ belongs to the transduction computed
by the $\iat$ if $w$ is accepted, the communication cell halts, and 
$v$ is the total output emitted during the computation.
The mappings $t_i$ and $t_o$ for the time complexities are defined analogously
to $\cat_{t_i,t_o}$. If $t_i$ or $t_o$ is the function $n+1$, we call it \emph{real time} and write $rt$. 
If $t_i(n)$ or $t_o(n)$ is of the form $k\cdot n$, for some 
$k\in\mathbb{Q}$, $k\geq 1$, we call it \emph{linear time} and write $lt$.

Any transduction computed by a cellular automaton can be divided into two 
tasks. One is the acceptance of the input, the other
one the transformation of the input into the output. Both tasks have to end
successfully in order to obtain a valid computation. On the one hand, 
this allows to modularize constructions of cellular automaton transducers
as both parts can be implemented independently on different tracks. On the
other hand, this implies that a language, which is not accepted by any 
cellular automaton in time $t_i$, cannot be the projection on the 
first components of any transduction belonging to any class 
$\tfam(\cat_{t_i,t_o})$. Unfortunately, it is a long-standing open
problem whether there are languages accepted by two-way cellular automata
in arbitrary, that is, exponential time but cannot be accepted in 
real time~(see, for example,~\cite{kutrib:2009:calt}). However, the same observation applies to
transductions computed by iterative arrays. So,
we obtain the following theorem.

\begin{theorem}\label{theo:notin-iatrt}
Let $t_o:\mathbb{N}\to\mathbb{N}$, $n\leq t_o(n)$, be a time complexity.
Then there exists a language belonging to $\tfam(\cat_{rt,rt})$, but not
to $\tfam(\iat_{rt,t_o})$.
\end{theorem}

\begin{proof}
The language 
$$
L=\{\,\kand x_k \kand \cdots\kand x_1 \dollar y_1\kand \cdots\kand
y_k\kand \mid k\geq 1, x_i^R=y_iz_i \mbox{ and } x_i, y_i, z_i \in \{a,b\}^*\,\}
$$
is not accepted by
any real-time~$\ia$~\cite{kutrib:2009:calt}.
However, it is linear context free. Since all linear
context-free languages are accepted by one-way cellular automata
in real time~\cite{Smith:1970:CAFL}, the transduction
\mbox{$\{\,(w, a^{|w|}) \mid w\in L\,\}$} is a witness for the assertion.
\end{proof}

The previous theorem shows that there are transductions which cannot be computed by 
any iterative array that has to accept the input in real time. In fact,
the limitation arises from the limitation to accept languages. This raises
the question whether there are witness transductions whose 
projections on the first components are accepted by real-time iterative
arrays, that is, the limitation is a limitation to transform the input
in time. The next example answers the question for real time in the affirmative.

\begin{example}\label{exa:mirror-language}
In~\cite{kutrib:2010:tcia:proc} it has been shown that the
transduction $\{\, (w, w^R)\mid w\in\{a,b\}^*\,\}$ does not
belong to the family $\tfam(\iatrtrt)$. However, it can be computed
by a $\cat_{rt,rt}$ $M$
as follows. Transducer $M$ performs two tasks on different tracks
in parallel. The first one is to synchronize the cells in
real time as has been shown in Example~\ref{exa:sort}. 
The second task is to reverse the input in real time. So,
when the cells fire they emit their current input symbol
in order to complete the transduction.

The second task is computed by a cellular automaton 
$M'=\langle S,F,A,B,\border,\delta\rangle$ 
that itself uses two tracks which are implemented by the state set
$S= (A\cup\{\lambda\})^2$ (see~Figure~\ref{fig:pic-mirr}).
Let $(p_1,q_1)$, $(p_2,q_2)$, and $(p_3,q_3)$ be arbitrary states
from $S$. Then
$$
\begin{array}{rcl}
\delta(\border,(p_1,q_1),\border) &=& (q_1,p_1),\\
\delta(\border,(p_1,q_1),(p_2,q_2)) &=& (p_2,p_1),\\
\delta((p_1,q_1),(p_2,q_2),\border) &=& (q_2,q_1), \text{ and}\\
\delta((p_1,q_1),(p_2,q_2),(p_3,q_3)) &=& (p_3,q_1)
\end{array}
$$
shift the contents of the upper track to the left and
the contents of the lower track to the right. Symbols arriving at the left end
are copied to the lower track, and symbols arriving at the right
end are copied to the upper track. In this way the input
circulates. If $M'$ is started with $w\in A^+$ on its
upper track and empty lower track, then the reversal of
$w$ is written on the lower track after in $|w|$ time
steps, that is, when the FSSP of the first task fires.
\qed
\end{example}

\begin{figure}[!ht]
\centering
\includegraphics[scale=.8]{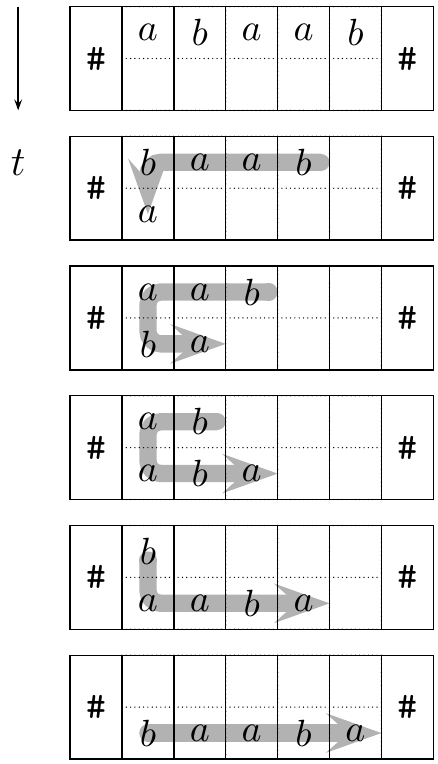}
\caption{Space-time diagram of a two-way cellular automaton reversing its input.}
\label{fig:pic-mirr}
\end{figure}

Theorem~\ref{theo:notin-iatrt} left open whether there is a proper inclusion
between the transduction families or whether they are incomparable.
Though real-time two-way cellular automata accept a strictly larger family
of languages than real-time iterative arrays, we cannot conclude
that there is an inclusion between the transduction families.
The reason is that iterative arrays receive their input sequentially 
to the communication cell \emph{and} emit the output sequentially
by the communication cell as well. So, when the last input symbol
is read by the leftmost $\iat$ cell, the last output is emitted 
also by the leftmost cell. However, for $\cat$ this last output
has to be at the right of the remaining output. Therefore, the
usual simulation of an $\ia$ by a $\ca$ where the leftmost $\ca$ cell
simulates the communication cell and the input is successively
shifted to the left does not work. Nevertheless, the next result
shows that the simulation is possible as long as `fair' time complexities
are considered. Clearly, an iterative array transducer can emit
an output up to the time complexity $t_o$, while in a cellular automaton
transducer each cell can emit only one output. So,
any time complexity $t_o$ larger than linear time  
yields a trivial transduction computed by an $\iat$ but
not by any $\cat$ and, from this point of view,
is `unfair'.

\begin{theorem}\label{theo:iat-by-cat-simulation}
Let $t_i,t_o:\mathbb{N}\to\mathbb{N}$ be two mappings so that
$t_i(n)\leq k_1\cdot n$, $t_o(n)\leq k_2\cdot n$, for two constants
$1\leq k_1,k_2$.
Then any transduction belonging to $\tfam(\iat_{t_i,t_o})$ is computed
by some $\cat_{t_i,t_o}$.
\end{theorem}

\begin{proof}
Let $M=\langle S, F, A, B, \rightend, s_0, \delta,\delta_0 \rangle$
be an $\iat_{t_i,t_o}$. By standard techniques $M$ can be modified such that
a cell never reenters the quiescent state after having left it, and
that never more than $n$ cells are non-quiescent on inputs of length $n$
until the transduction is completed. The former property can be achieved 
by introducing a new state to which non-quiescent cells change instead
of the quiescent state. The latter property is obtained by grouping
$\max\{k_1,k_2\}$ cells into one.

As mentioned above, the transduction computed by~$M$ can be divided 
into two tasks running on different tracks. Since any language accepted 
by a \mbox{$t_i$-time} iterative array is known
to be accepted by a $t_i$-time two-way cellular automaton as well,
it remains to be shown how to simulate the transformation of the
input into the output by a $\cat$ $M'=\langle S',F',A,B,\border,\delta'\rangle$.

Assume for a moment that $k_2=1$, that is, $t_o$ is real time.
Basically, the idea of the simulation (of the second task) is as follows
(see~Figures~\ref{fig:pic-regs} and~\ref{fig:pic-tssim}).
Every cell of $M'$ has five registers. In the first and second register,
cells of~$M$ are simulated. So, they initially carry the quiescent state of $M$. 
At the beginning, the leftmost cell of~$M'$ simulates the communication
cell of $M$ for two time steps. Then the second cell of $M'$ simulates the
communication cell for another two time steps, and so on. When cell $i$
simulates the communication cell, then the concatenation of the
first two registers of cells $i, i-1,\dots, 1$ represent the states of the
cells $1,2,\dots, 2i-1$ or $1,2,\dots, 2i$ of $M$. In order to provide
the necessary input symbols for the simulation of the
communication cell, on the third track $M'$ shifts its input to the left
at every other time step. To this end, a modulo two counter is maintained
in the fourth registers. When the end-of-input symbol meets
the simulation of the communication cell in cell $\lceil n/2\rceil$,
the simulation of $M$ is completed.
To conclude the idea of the construction the output has to be described.
In the right half of the automaton, each cell passed through by the
end-of-input symbol emits $\lambda$. In the left half, each cell
emits its output when it has simulated two steps of the communication
cell. The output is the concatenation of the two outputs generated
by the communication cell. To this end, the first output has to be
remembered for one time step in the fifth register. Dependent on the parity
of the length of the input, the last output possibly has to be emitted by a
cell having simulated only one step of the communication cell. 
So, the output of the $\iat$ is simulated in the left half of the $\cat$
while the right half actually emits the empty word.

Formally, the construction is as follows 
(see~Figures~\ref{fig:pic-regs} and~\ref{fig:pic-tssim}).
$$
S'= S\times S \times (A\cup\{\rightend\})\times\{0,1\}\times B^{\leq j},
$$
where $j$ is the length of the longest output emitted in one step by the
communication cell of $M$, and $\rightend$ is the end-of-input symbol of the $\iat$. 
On input $w=a_1a_2\cdots a_n \in A^n$, 
cell $i$ is initially in state $(s_0, s_0, a_i, 0, \lambda)$.

\begin{figure}[!ht]
\centering
\includegraphics[scale=.9]{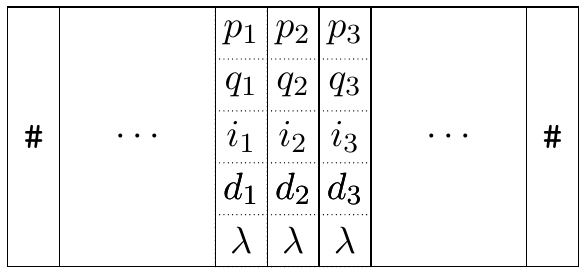}
\caption{Structure of $\cat$ registers and denotation of their contents.
Here all fifth registers are empty.} 
\label{fig:pic-regs}
\end{figure}

\begin{figure}[!ht]
\centering
\includegraphics[scale=.9]{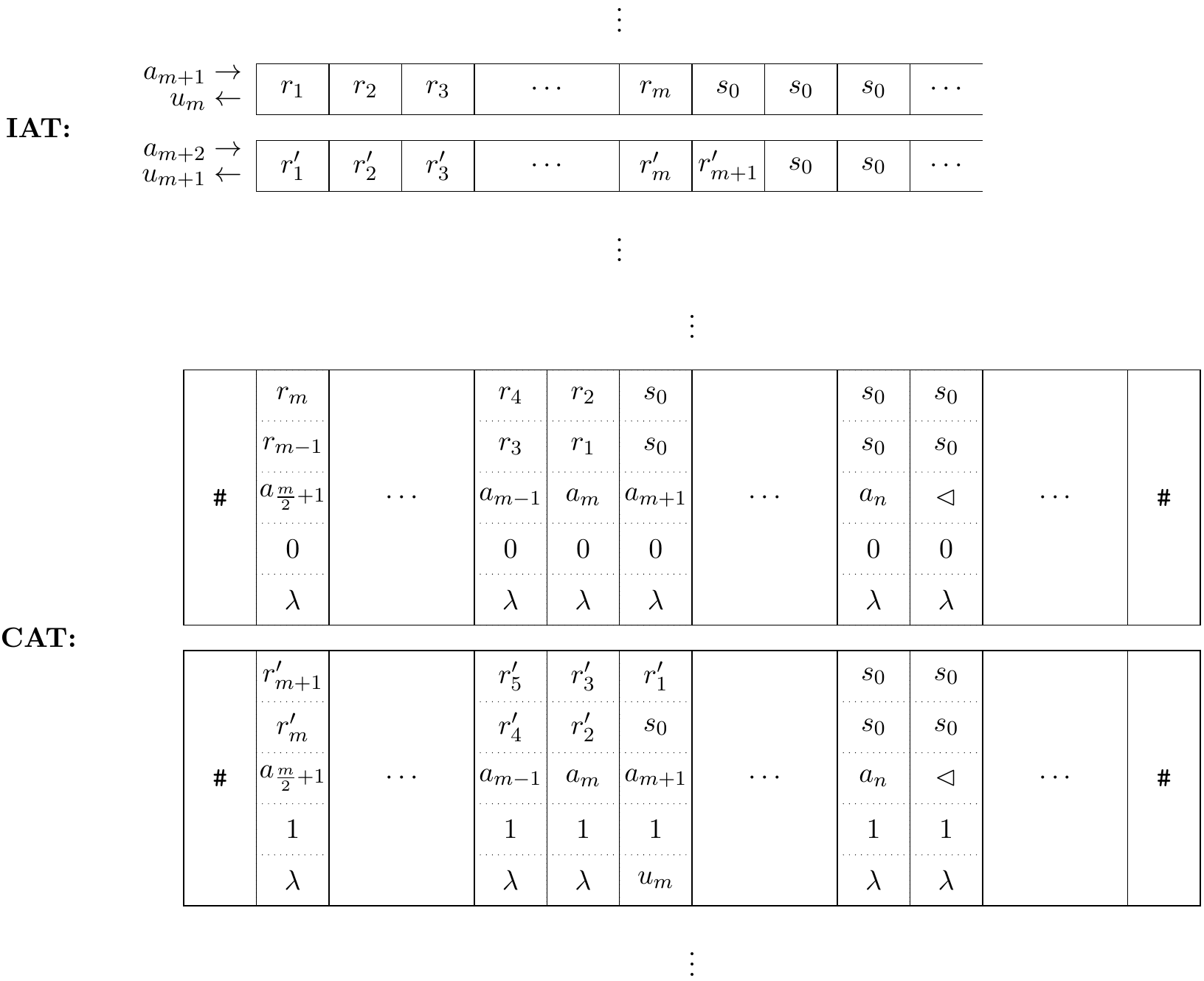}
\caption{Principle of a $\cat$ simulating an $\iat$. The input 
is $a_1a_2\cdots a_n$ and $m$ is even. Depicted are the two
consecutive configurations at time steps $m$ and $m+1$
for the $\iat$ and for the $\cat$.} 
\label{fig:pic-tssim}
\end{figure}

A cell that detects that its left neighbor has just filled the first
two registers, starts to simulate the communication cell
for two time steps. Similarly, so does the leftmost cell
at initial time. As before, $\delta_{0,s}$~($\delta_s$) denotes
the first component of the value of $\delta_{0}$ ($\delta$), 
while $\delta_{0,o}$ ($\delta_o$) denotes the second
component. So, for $p_1,p_2,q_1 \neq s_0$,
we define
\begin{equation*}
\begin{split}
\delta'(\border, (s_0, s_0, i_2, 0, \lambda), (s_0, s_0, i_3, 0, \lambda))
 &= ((\delta_{0,s}(i_2, s_0, s_0), s_0, i_2, 1,\delta_{0,o}(i_2, s_0, s_0)), \lambda),\\
\delta'(\border, (p_2, s_0, i_2, 1, \beta), (s_0, s_0, i_3, 1, \lambda))
 &=\\
 ((\delta_{s}(p_2, s_0, s_0), & \delta_{0,s}(i_3, p_2, s_0), i_3,0,\lambda), 
                                           \beta\delta_{0,o}(i_3, p_2, s_0)),\\[2mm]
\delta'((p_1,q_1,i_1,0,\lambda), (s_0, s_0, i_2, 0, \lambda), (s_0, s_0, i_3, 0, \lambda))
 &= ((\delta_{0,s}(i_2, q_1, p_1), s_0, i_2, 1,\delta_{0,o}(i_2, q_1, p_1)), \lambda),\\
\delta'((p_1,q_1,i_1,1,\lambda), (p_2, s_0, i_2, 1, \beta), (s_0, s_0, i_3, 1, \lambda))
 &=\\
 ((\delta_{s}(p_2, q_1, p_1), & \delta_{0,s}(i_3, p_2, q_1), i_3,0,\lambda), 
                                              \beta\delta_{0,o}(i_3, p_2,
                                              q_1)).
\end{split}
\end{equation*}
A cell that already has simulated two steps of the communication cell
continues to simulate cells of the $\iat$ by applying the following
transitions, where $p_1,q_1, p_2,q_2 \neq s_0$:
\begin{equation*}
\begin{split}
\delta'(\border, (p_2, q_2, i_2, d_2, \lambda), (p_3, q_3, i_3, d_3, \lambda))
 &=\\
 ((\delta_{s}(p_2, s_0, s_0), & \delta_s(q_2,p_2,s_0), i_{2+d_2}, 1-d_2,\lambda), \lambda),\\
\delta'((p_1, q_1, i_1, d_1, \lambda), (p_2, q_2, i_2, d_2, \lambda), (p_3, q_3, i_3, d_3, \lambda))
 &=\\
 ((\delta_{s}(p_2, q_1, p_1), & \delta_s(q_2,p_2,q_1), i_{2+d_2}, 1-d_2,\lambda), \lambda).
\end{split}
\end{equation*}
Finally, the cells that did not simulate a step of the
communication cell behave according to the following transitions, where
$p_1\ne s_0$:
\begin{align*}
\delta'((p_1, s_0, i_1, 1, \lambda), (s_0, s_0, i_2, 1, \lambda), (s_0, s_0, i_3, 1, \lambda))
 &= ((s_0, s_0, i_3, 0,\lambda), \lambda),\\
\delta'((s_0, s_0, i_1, 0, \lambda), (s_0, s_0, i_2, 0, \lambda), (s_0, s_0, i_3, 0, \lambda))
 &= ((s_0, s_0, i_2, 1,\lambda), \lambda),\\[2mm]
\delta'((p_1, s_0, i_1, 1, \lambda), (s_0, s_0, i_2, 1, \lambda), \border)
 &= ((s_0, s_0, \rightend, 0,\lambda), \lambda),\\
\delta'((s_0, s_0, i_1, 0, \lambda), (s_0, s_0, i_2, 0, \lambda), \border)
 &= ((s_0, s_0, i_2, 1,\lambda), \lambda).
\end{align*}
This concludes the construction for the assumption $k_2=1$, that is,
for $t_o(n)=n$. Now let constant $k_2$ be at least two. In this case, the simulation
is slightly modified as follows. Each cell simulates successively~$k_2$ steps of the
communication cell. In order to provide the correct input symbols the input
is shifted to the left $k_2-1$ times within $k_2$ steps. To this end,
a modulo $k_2$ counter is maintained. The simulation is completed when
the rightmost cell (cell $n$) of the $\cat$ has finished to simulate the communication
cell. Clearly, at that time $k_2 \cdot n$ steps have been simulated.
Similar to the construction above, a cell emits its output after having
finished to simulate the communication cell, and the output is the
concatenation of the outputs computed during these $k_2$ steps.
\end{proof}

So, for the time complexities real time and linear time, the parallel input/output mode
is not weaker than the sequential one. 
In fact, Theorems~\ref{theo:notin-iatrt} and~\ref{theo:iat-by-cat-simulation}
imply that the former is strictly stronger for $(rt,rt)$ and $(rt,lt)$:

\begin{corollary}\label{cor:incl-rt}
The family $\tfam(\iatrtrt)$ is strictly included in $\tfam(\catrtrt)$,
and $\tfam(\iatrtlt)$ is strictly included in $\tfam(\catrtlt)$.
\end{corollary}

Since the families of languages accepted by two-way
cellular automata and iterative arrays in linear time are known to be identical,
the questions for the precise relations between 
$\tfam(\iatltrt)$ and $\tfam(\catltrt)$ or between $\tfam(\iatltlt)$ and
$\tfam(\catltlt)$ raise immediately.

\begin{proposition}
The family $\tfam(\iatltrt)$ is strictly included in $\tfam(\catltrt)$.
\end{proposition}

\begin{proof}
The inclusion $\tfam(\iatltrt)\subseteq \tfam(\catltrt)$ follows again by 
Theorem~\ref{theo:iat-by-cat-simulation}.
Moreover, if the transduction $\{\, (w, w^R)\mid w\in\{a,b\}^*\,\}$
would be computable by some $\iatltrt$, then it would be computed
by an $\iatrtrt$, since the trivial input to be accepted is
$\{a,b\}^*$. However, by Example~\ref{exa:mirror-language}
this language separates the families $\tfam(\iatrtrt)$
and $\tfam(\catrtrt)$ and, thus, it separates
$\tfam(\iatltrt)$ and $\tfam(\catltrt)$.
\end{proof}

For the last time complexity in question $(lt,lt)$ we obtain a
different situation. The parallel and sequential input/output
modes are equally powerful. 

\begin{theorem}\label{theo:ltlt-equal}
The families $\tfam(\iatltlt)$ and $\tfam(\catltlt)$ are identical.
\end{theorem}

\begin{proof}
The inclusion $\tfam(\iatltlt)\subseteq \tfam(\catltlt)$ follows once more by 
Theorem~\ref{theo:iat-by-cat-simulation}.

Conversely, an $\iatltlt$ can simulate a $\catltlt$
as follows. In a first phase, it reads 
the input and stores it successively in its cells. 
In a second phase, the iterative array transducer starts a
FSSP in the communication cell, that synchronizes the $n$ cells 
within $2n-2$ time steps. Finally, all cells start the 
simulation of the $\cat$ at the same time. Clearly, the iterative 
array transducer obeys linear time bounds if the cellular automaton
transducer does.
\end{proof}

The previous result can be generalized to arbitrary time complexities
beyond linear time as long as the iterative arrays use linear space only.
For space complexities beyond linear space, clearly, iterative arrays
are stronger than cellular automata, since the latter are linearly 
space bounded by definition.  

\section{Comparison with Finite State Transducers and Pushdown Transducers}
\label{sec:comptransd}

Here, we turn to compare cellular automaton transducers with finite state transducers (FST)
and pushdown transducers (PDT). These devices are in essence finite automata and pushdown
automata, where each transition is associated with a possibly empty output
word (see~\cite{Aho:1972:tptcv1}).
In their most general form, FST and PDT are nondeterministic devices, that is, 
the partial transition function of an FST maps from 
\mbox{$S \times (A \cup \{\lambda\})$} into the finite subsets of
$S \times B^*$. As above, $S$ denotes the state set and $A$ the input alphabet.
The partial transition function of a PDT maps from 
$S \times (A \cup \{\lambda\}) \times G$ into the finite subsets of
$S \times B^* \times G^*$, where $G$ denotes the pushdown alphabet.
Since a nondeterministic transducer may transform an input into different 
outputs, which is impossible for deterministic $\cat$, in the sequel we only 
study deterministic, unambiguous, and single valued devices. 

An FST $M$ is called \emph{single valued} ($\sfst$) if for all $(w_1,v_1), (w_2,v_2) \in T(M)$
either $(w_1,v_1)= (w_2,v_2)$ or $w_1\neq w_2$. 
An $\sfst$ is said to be \emph{unambiguous} ($\ufst$) if for all $(w,v) \in T(M)$
there is a unique computation transforming $w$ into $v$. Finally, a $\ufst$ is \emph{deterministic} ($\dfst$) 
if any computation is deterministic.  
It has been shown in~\cite{Weber:1995:edsvt} that every single-valued finite state 
transducer can be simulated by an unambiguous one. 
Furthermore, it is known~(see, for example,~\cite{Weber:1995:edsvt}) that 
$$
\tfam(\dfst)\subset \tfam(\ufst)=\tfam(\sfst).
$$

The notions of single-valued PDT ($\spdt$), unambiguous PDT ($\updt$), 
and deterministic PDT ($\dpdt$) are defined analogously.
Additionally, a $\updt$ is called \emph{real-time deterministic} ($\dpdt_\lambda$)
if it is not allowed to move on empty input. The following proper hierarchy
is known: (see, for example,~\cite{kutrib:2010:tcia:proc})
$$
\tfam(\dpdt_\lambda)\subset \tfam(\dpdt)\subset \tfam(\updt)\subset
\tfam(\spdt).
$$

In~\cite{kutrib:2010:tcia:proc} it has been shown 
that any $\dfst$ can be simulated by some $\iatrtrt$,
and any $\sfst$ can be simulated by some $\iatrtlt$. Here, we prove that both devices can be
simulated by some $\cat$ as well. 
Interestingly, the device with parallel input/output mode can compute the
transductions fast, in particular in real time, 
which is in contrast to the devices with sequential input/output mode.

\begin{lemma}\label{lem:sfst}
The families $\tfam(\dfst)$ and $\tfam(\sfst)$ are strictly included in $\tfam(\catrtrt)$.
\end{lemma}

\begin{proof}
We consider the transduction $\{\,(w,w^R) \mid u \in \{a,b\}^*\,\}$ of
Example~\ref{exa:mirror-language} which belongs to $\tfam(\catrtrt)$,
but clearly cannot be computed by any finite state transducer.
Next, we describe how a $\catrtrt$ can simulate an $\sfst$. Trivially, 
this construction applies to $\dfst$ as well. The idea of the simulation 
is similar to the construction for $\iatrtlt$ given
in~\cite{kutrib:2010:tcia:proc}. However, here we can reduce the time
complexity and, thus, have to cope with the problem of speeding up the
computation to real time. 

Let $M=\langle S, F, A, B, s_0, \delta \rangle$ be an unambiguous $\sfst$. 
Due to a result in~\cite{Weber:1995:edsvt} we may assume without loss
of generality that $M$ does not move on empty input. 
First, from $M$ a nondeterministic finite automaton
$M_{NFA}=\langle S, F, A, s_0, \delta'\rangle$ is extracted that
accepts $L(M)$. Then, automaton $M_{NFA}$ is converted into an equivalent
deterministic finite automaton $M_{DFA}$ by the powerset construction.

Now we turn to the construction of a $\catrtrt$ $M'$ which simulates $M$.
Transducer $M'$ has several tracks. The input is stored in the
first register and, additionally, its second register is used 
to shift the input to the left in every time step.
In the third register of the leftmost cell the deterministic finite automaton 
$M_{DFA}$ is simulated which receives its input on the second track.
Now $M'$ accepts if and only if $M_{DFA}$ accepts the input $w$ at time $|w|$.

The second task of $M'$ is to compute the output. 
For this purpose, the unique accepting computation of $M_{NFA}$ has to be identified
among all computations on $w$. As a first step, automaton $M_{NFA}$ is converted
into an equivalent right linear grammar $G_{NFA}$ with axiom $X$.
The productions of $G_{NFA}$ have three different forms:
\begin{enumerate}
\item 
$X \rightarrow a[q']$ for all transitions $q' \in \delta'(s_0,a)$
with $a \in A$,
\item 
$[q] \rightarrow a[q']$ for all transitions $q' \in \delta'(q,a)$ 
with $q \in S, a \in A$,
\item 
$[q] \rightarrow a$ for all transitions $q' \in \delta'(q,a)$
with $q \in S$, $q' \in F$, and $a \in A$.
\end{enumerate}
So, every production in $G_{NFA}$ corresponds to a transition rule in
$M_{NFA}$ and $M$ and, thus, corresponds to an output $u \in B^*$.

Let $w=a_1a_2 \cdots a_n$. We consider sets $V_1, V_2, \ldots, V_n$
of nonterminals from $G_{NFA}$ so that $Y \in V_i$, if and only if there is a derivation
$Y \Rightarrow^* a_{i} a_{i+1} \cdots a_{n}$ in $G_{NFA}$. 
Set $V_n$ includes exactly all nonterminals $Y$ for which the production 
$Y \rightarrow a_n$ belongs to $G_{NFA}$. In general, for $1 \le i <n$, the set 
$V_i$ includes exactly all nonterminals $Y$ for which the production 
$Y \rightarrow a_iZ$ belongs to $G_{NFA}$ and $Z\in V_{i+1}$.
Clearly, set $V_i$ can be computed from $a_i$ and $V_{i+1}$.

Let us assume for a moment that $n$ is even.
The next construction step is to set up $M'$ so that $V_1, V_2, \ldots, V_n$ are computed
in the cells $\frac{n}{2}+1, \frac{n}{2}+2, \ldots, n$
within $\frac{n}{2}$ time steps. To this end, on an additional track 
the input is shifted to the right in every time step. Moreover,
in the first step the rightmost cell computes the sets
$V_{n-1}$ and $V_n$ with the knowledge of $a_{n-1}$ and $a_n$.
In the next time step, cell $n-1$ computes the sets $V_{n-3}$ and $V_{n-2}$
with the knowledge of $a_{n-3}$, $a_{n-2}$, $V_{n-1}$, and $V_n$.
In general, cell $n-i+1$ computes the sets $V_{n-2(i-1)-1}$ and $V_{n-2(i-1)}$
in time step $i$ with the knowledge of $a_{n-2(i-1)-1}$, $a_{n-2(i-1)}$, 
$V_{n-2(i-1)+1}$, and $V_{n-2(i-1)+2}$.
Additionally, the symbols $a_{n-2(i-1)-1}$ and $a_{n-2(i-1)}$ are stored
in another two registers. Thus, at time step
$\frac{n}{2}$ the sets $V_1$ and $V_2$ are computed in cell $\frac{n}{2}+1$.
(see~Figure~\ref{fig:sfst:1} for an example). 
The case when $n$ is odd is handled similarly. Then, 
the sets $V_1, V_2, \ldots, V_n$ are computed in the cells 
$\lceil \frac{n}{2} \rceil, \lceil \frac{n}{2} \rceil+1, \ldots, n$
within $\lceil \frac{n}{2} \rceil$ time steps.

From the sets $V_i$ now the unique accepting computation path of $M$ is
extracted. Clearly, $w\in L(G_{NFA})$ if and only if $X\in V_1$. Moreover,
there is only one production of the form $X\to a_1Z_1$ in $V_1$. Otherwise
the accepting path would not be unique. For the same reason
there is only one production of the form $Z_1\to a_2Z_2$ in $V_2$, and so on.
Let us again assume for a moment that $n$ is even, and let
the unique sequence of productions that derive $w$ be $p_1,p_2,\dots,p_n$.
At time step $\frac{n}{2}+1$, the cells $\frac{n}{2}$ and $\frac{n}{2}+1$
determine the productions $p_1$ and $p_2$. This is possible, since 
both cells can identify themselves in time step $\frac{n}{2}+1$,
and all necessary information is available in the cells and their 
neighborhoods. Additionally, the productions computed are sent to the left 
on an additional track. Furthermore, in cell $\frac{n}{2}+1$ a signal $R$ 
is sent to the right. In the next time step, 
cell $\frac{n}{2}+1$ determines $p_3$ and by signal $R$
cell $\frac{n}{2}+2$ is caused to determine $p_4$. 
Both productions are subsequently shifted to the left.
In general, $R$ arrives at cell $\frac{n}{2}+i$ at time step $\frac{n}{2}+i$ and causes
cell $\frac{n}{2}+i-1$ to determine $p_{2i-1}$ and cell $\frac{n}{2}+i$ to determine $p_{2i}$. 
Again, all necessary information is available in the cells and their neighborhoods. 
The case when $n$ is odd can be handled
similarly again. In this case, the productions $p_1, p_2, \ldots, p_n$ are computed in the same way
in the cells $\lceil \frac{n}{2} \rceil-1, \lceil \frac{n}{2} \rceil, \ldots,
n$ and similarly are sent to the left.

By construction, production $p_i$ reaches cell $i$ at time step $n$, for $1 \le i \le n$.
At this moment, the output $u_i \in B^*$ associated with the transition rule
that led to the definition of the production has to be emitted in cell $i$. 
So, it remains to be ensured that all cells are synchronized at time step $n$. 
As is described in Example~\ref{exa:sort}, this can be achieved by
simulating an FSSP on another track.

Altogether, we obtain that $M'$ simulates $M$, accepts and emits the output in
real time. 
Thus, $\sfst$ $M$ is simulated by a $\catrtrt$.
\end{proof}

\begin{figure}[!ht]
\centering
\includegraphics[scale=.9]{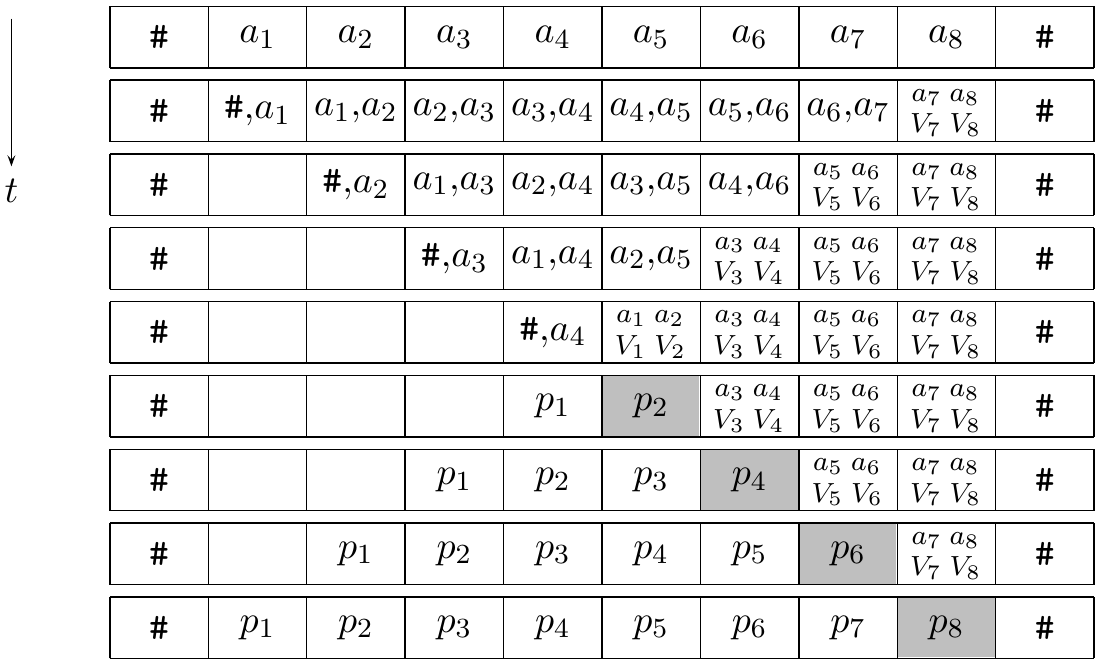}
\caption{Schematic computation of a $\catrtrt$ simulating an $\sfst$ on input $a_1a_2 \cdots a_8$.
In the first four time steps, the sets $V_1,V_2, \ldots, V_8$ are computed in cells $5$, $6$, $7$, and
$8$. Thus, an accepting path has been stored in the last four cells which is extracted in the last four time
steps and the corresponding output of the transitions is distributed to the correct cells.
The simulation of $M_{DFA}$ and the synchronization is not depicted.}
\label{fig:sfst:1}
\end{figure}

The next result follows from known results on $\iatrtrt$ and the simulation of $\iat$ by
$\cat$  as presented in Section~\ref{sec:compcap}. It is worth mentioning that 
$\tfam(\sfst)$ and $\tfam(\dpdt_\lambda)$ are
incomparable~\cite{kutrib:2010:tcia:proc}. Here, we obtain that both classes
are included in $\tfam(\catrtrt)$. 

\begin{lemma}
The family $\tfam(\dpdt_\lambda)$ is strictly included in $\tfam(\catrtrt)$.
\end{lemma}

\begin{proof}
The assertion follows from the fact that 
$\tfam(\dpdt_\lambda)$ is strictly included in $\tfam(\iatrtrt)$
(\cite{kutrib:2010:tcia:proc}) and that $\tfam(\iatrtrt)$ is strictly
included in $\tfam(\catrtrt)$ due to Corollary~\ref{cor:incl-rt}.
\end{proof}

Finally, we will show that any $\dpdt$ can be simulated by some $\catrtlt$. This is again an
improvement in comparison with $\iat$. It has been shown in~\cite{kutrib:2010:tcia:proc} that
any $\dpdt$ can be simulated by some $\iatltlt$ which in turn can be simulated by some $\catltlt$
owing to Theorem~\ref{theo:ltlt-equal}. Here, we obtain that the simulation
can already be achieved
by some $\catrtlt$.

\begin{lemma}\label{lem:dpdt}
The family $\tfam(\dpdt)$ is strictly included in $\tfam(\catrtlt)$.
\end{lemma}

\begin{proof}
The transduction $T=\{\,(ww,wc^{|w|}) \mid w \in \{a,b\}^+\,\}$ cannot
be computed by any pushdown transducer, since the language $\{\,ww \mid w \in
\{a,b\}^+\,\}$ is not context free.

On the other hand, transduction $T$ can be computed by a $\catrtrt$
and, thus, by a $\catrtlt$. To this end, as in Example~\ref{exa:sort}
two instances of the FSSP are initiated at both ends of the array
that cause each cell to fire at time $n$. Firing of a cell in the left 
half means to emit the original input symbol and 
firing in the second half means to emit symbol $c$. Since
the language $\{\,ww \mid w \in \{a,b\}^+\,\}$ is accepted by a real-time 
$\ca$ this shows $T\in \tfam(\catrtrt)$.

Given a deterministic pushdown transducer $M$
with state set $S$ and pushdown store alphabet $G$, we 
next construct a $\catrtlt$ $M'$ simulating $M$.

There is a constant $k_1\geq 0$ such that $M$ cannot push more than $k_1$ 
pushdown symbols in one time step, and $M$ can pop at most one pushdown 
symbol in one time step.
Moreover, there is a constant $k_2 \le |S| \cdot |G|$ such that 
$M$ cannot perform more than $k_2$ subsequent moves on empty input.
{F}rom these facts follows that $M$ works in linear time. Let $k=\max\{k_1,k_2\}$.

Basically, the $\catrtlt$ $M'$ computes five tasks on different tracks. 
On the first track, a deterministic pushdown automaton is simulated in real
time that accepts the language $L(M)$. The details of such a simulation
can be found in~\cite{kutrib:2008:ca-cpv}.

It remains to be shown how $M'$ computes the output of $M$.
The second track is used as follows. The leftmost cell simulates
the state transitions of $M$ while on request the other cells shift the input
to the left, thus, providing the input for the leftmost cell. 
In detail, when the simulation of $M$ consumes an input symbol, a signal 
is sent to the right which causes the cells to shift the input one position 
to the left. Otherwise, when $M$ simulates a transition on empty input 
no signal is sent.

The pushdown store of $M$ is simulated on the third track.
In~\cite{kutrib:2008:ca-cpv} it has been shown how to simulate
the data structure pushdown store without loss of time.
Since here at most $k$ symbols are pushed in one time step, 
the simulation can be realized by grouping $k$ pushdown symbols together.
  
On the fourth track, a data structure queue is implemented as is also shown
in~\cite{kutrib:2008:ca-cpv}. The leftmost cell stores transitions
simulated on the second track into this queue. Since at most $k$ 
consecutive transitions are on empty input, grouping at most $k+1$ transitions
into one symbol to be stored ensures that any symbol in the queue 
represents at least one transition consuming an input symbol.
So, at most as many symbols 
are stored as the input is long.
Since $M$ works in linear time, it is not difficult to see that all these tasks
are simulated by $M'$ in linear time as well. 

The final task is to emit the output. After acceptance
of the input, on the fifth track a signal is started from the leftmost cell to the right
which provides sufficient time so that all symbols are properly stored in the
queue. Having reached the rightmost cell, the signal changes its direction 
and moves back to the leftmost cell. On its way it causes each cell
passed through to emit $u_1 u_2 \cdots u_m \in B^*$, 
if it stores $(p_1, p_2, \ldots, p_m)$ in its fourth register (the queue
register), where $u_i$ is the output associated with transition~$p_i$, $1\leq
i\leq m$. If the fourth register is empty, the cell emits $\lambda$.
When the signal arrives at the leftmost cell again, 
the transduction is completed in linear time.
\end{proof}

\end{document}